\documentclass[12pt]{amsart}

\usepackage[utf8]{inputenc}
\usepackage[english]{babel}
\usepackage{amsmath}
\usepackage{amsfonts}
\usepackage{amssymb}
\usepackage{graphicx}
\usepackage{mathtools}
\usepackage{amsthm}

\newtheorem{thm}{Theorem}[section]
\newtheorem{prop}[thm]{Proposition}
\newtheorem{lem}[thm]{Lemma}

\newtheorem{rem}[thm]{Remark}

\theoremstyle{definition}
\newtheorem{defi}{Definition}[section]
\newcommand{\R}{\mathbb{R}}
\renewcommand{\P}{\mathbb{P}}

\newcommand{\Z}{\mathbb{Z}}
\newcommand{\C}{\mathbb{C}}
\newcommand{\Y}{\mathbb{Y}}

\newcommand{\E}{\mathbb{E}}

\newcommand{\T}{\mathbb{T}}

\usepackage{fullpage}

\usepackage{hyperref}

\title{The 2D Toda lattice hierarchy for multiplicative statistics of Schur measures}
\author[P. Lazag]{Pierre Lazag }

\date{}
\begin{document}

\begin{abstract}
    We prove that Fredholm determinants built from generalizations of Schur measures, or equivalently, arbitrary multiplicative statistics of the original Schur measures, are tau-functions of the 2D Toda lattice hierarchy. Our result applies to finite temperature Schur measures and extends both the results of Okounkov in \cite{okounkovschurmeasures} and of Cafasso--Ruzza in \cite{cafassoruzza} concerning the finite temperature Plancherel measure. Our proof relies on the semi-infinite wedge formalism and the Boson-Fermion correspondence.
\end{abstract}
\maketitle

\address{SISSA, via Bonomea 265, 34136, Trieste, Italy

LAREMA, UMR CNRS 6093, 2 Boulevard Lavoisier
49045 Angers cedex 01, France

Institut de Recherche en Math\'ematique et Physique, UCLouvain, 

Chemin du Cyclotron 2, 1348 Louvain-la-Neuve, Belgium}

\email pierrelazag@hotmail.fr 

\normalsize

\section{Introduction}
The goal of this note is to build tau-functions for the 2D Toda lattice hierarchy out of natural generalizations of Schur measures or, equivalently, of multiplicative functionals of Schur measures. Schur measures were introduced by Okounkov in \cite{okounkovschurmeasures} and are probability measures on the set of Young diagrams parametrized by two countable sets of parameters $t=(t_1,t_2,\dots)$, $t'=(t_1',t_2',\dots) \subset \C$ giving rise to determinantal point processes on $\Z + 1/2$ via the map $\lambda \to \{\lambda_i -i +1/2 \}$ with a correlation kernel $K_{t,t'}$. It is proved in \cite{okounkovschurmeasures}, Theorem 3, that the Fredholm determinants $\det(1- K_{t,t'})_{\ell^2\{n+1/2,n+3/2,\dots\}}$ are tau--functions of the 2D Toda lattice hierarchy, where the parameters $t$ and $t'$ play the role of the time parameters.

In the case when $t=t'=(L,0,0,\dots)$, $L>0$, the Schur measure is the poissonized Plancherel measure with parameter $L>0$, giving rise to the determinantal point process associated with the discrete Bessel kernel, see e.g. \cite{boo}, \cite{borodinbessel}. In \cite{cafassoruzza}, Theorem I, the authors derive integrable equations for Fredholm determinants of deformations of the Plancherel measure similar to the ones obtained in \cite{okounkovschurmeasures}, see equation (A.21) there and equation (1.11) in \cite{cafassoruzza}. Such deformations include the case of the finite temperature Bessel kernel as an example, see \cite{beteabouttier} and \cite{borodinperiodic}. Their proof relies on the fact that such Fredholm determinants may be seen as the expectation of a multiplicative functional under the original Plancherel measure, and then exploits the integrable structure of the discrete Bessel kernel and the related Riemann-Hilbert problem techniques.\\

In the present note, we consider the same deformation to the general Schur measures and prove that Fredholm determinants built out of these deformations are also tau-functions of the 2D Toda lattice hierarchy. Our result is both a generalization of those by \cite{cafassoruzza} and \cite{okounkovschurmeasures}. As in \cite{cafassoruzza}, we interpret these Fredholm determinants as expectations of multiplicative functionals under the original Schur measure, but our argument then differs from \cite{cafassoruzza} since the kernel under consideration may not be integrable. So, instead of using Riemann-Hilbert problem techniques, we use the semi-infinite wedge formalism introduced in \cite{okounkovschurmeasures} in the context of random partitions and describe our tau-functions in terms of matrix coefficients of operators on the fermionic Fock space. This approach to tau-functions for integrable hierarchies goes back to the so-called Kyoto school, see for example the book \cite{petitlivrebleu} or the survey \cite{alexandrovzabrodin} and references therein.\\

Examples include what we call finite temperature deformations of Schur measures, following the construction from \cite{beteabouttier} and \cite{borodinperiodic}. More concrete examples are the multi-critical Schur measures and their finite temperature versions, see \cite{beteabouttierwalsh} and also \cite{kimurazahabi}.

In another direction, it was also known from \cite{borodinpainleve}, \cite{baikrhp}, \cite{adlervanmoerbeke}, that specific choices for the parameters of the Schur measures, encompassing the Plancherel measure case, lead to discrete Painlev\'e equations for gap probabilities. Recently, in \cite{chouteautarricone}, the authors have extended these results to Schur measures with an arbitrary but finite number of parameters, obtaining a discrete Painlev\'e II hierarchy and recursion equations for the gap probabilities. Our result also applies to this particular choice of parameters they consider and for any multiplicative statistics. We wonder if one could find a deformed Painlev\'e hierarchy, extending the equations in Thm. III of \cite{cafassoruzza} by our methods. We also wonder if one could specialize our result to the case  of the homogeneous stochastic six-vertex model described in \cite{borodinsixvertex}, where it is shown that the Laplace transform of the height-function can be expressed as non-trivial multiplicative statistics of a Schur measure.\\

We start below by stating our main result precisely, Theorem \ref{thm:thm1} below, before describing the scheme of our argument in more detail.
\subsection{Statement of the result}
Let $t=(t_1,t_2,\dots)$, $t'=(t_1',t_2',\dots) \subset \C$ be two infinite sequences of numbers such that the series
\begin{align*}
    \sum_{n \geq 1}t_nz^n, \quad\sum_{n \geq 1}t_n'z^n
\end{align*}
define analytic functions in a neighborhood of the unit circle $\T:= \{z \in \C, \, |z|=1\}$, which holds if, for example, there exist $C>0$ and $r\in (0,1)$ such that $t_n,t_n' \leq Cr^n$ for all $n \geq 1$. We assume this condition holds throughout this text. We have in particular that
\begin{align*}
\sum_{n \geq 1} n |t_n t_n'| <+\infty,
\end{align*}
and we set
\begin{align} \label{eq:defZ}
Z_{t,t'}:= \exp \left( \sum_{n \geq 1} nt_n t_n' \right).
\end{align}
Observe that we have
\begin{align}
Z_{t,t'}=Z_{-t,-t'}=Z_{t',t}.
\end{align}
Define the functions
\begin{align*}
\gamma(z,t) &:= \exp \left( \sum_{ n \geq 1} t_n z^n \right),  \\
J(z;t,t')&:= \frac{\gamma(z,t)}{\gamma(z^{-1},t')} = \exp \left( \sum_{n \geq 1} t_nz^n -t_n' z^{-n} \right).
\end{align*}
We have the symmetries
\begin{align} \label{eq:J1}
J(z;t',t) = J(z^{-1},-t,-t')
\end{align}
and
\begin{align} \label{eq:J2}
J(z;-t,-t') = J(z;t,t')^{-1}
\end{align}
We define the sequence of complex numbers $ \left(J_k(t,t') \right)_{k \in \Z} $ via the generating Laurent series:
\begin{align*}
J(z;t,t')= \sum_{k \in \Z} J_k(t,t') z^k.
\end{align*}
From equations (\ref{eq:J1}) and (\ref{eq:J2}), we have
\begin{align} \label{eq:J3}
J_{-k}(-t,-t')= J_k(t',t)
\end{align}
and
\begin{align} \label{eq:J4}
\sum_{k \in \Z}J_k(-t,-t') z^k=  J(z;t;t')^{-1}. 
\end{align}
Let $\Z':= \Z +1/2$ be the set of half-integers and let $\sigma : \Z' \to [0,1]$ be a function such that
\begin{align} \label{cond:l1}
\sum_{ k \in \Z'_{<0}} \sigma(k) <+ \infty.
\end{align}
We form the kernel
\begin{align} \label{eq:defK}
K_{t,t',\sigma}(x,y) = \sum_{k \in \Z'} \sigma(k) J_{x +k}(t,t') J_{-y-k}(-t,-t'), \quad x,y \in \Z'.
\end{align}
From equations (\ref{eq:J3}) and (\ref{eq:J4}), the kernel $K_{t,t',\sigma}(x,y)$ is changed into $K_{t,t',\sigma}(y,x)$ under the transformation $ t \leftrightarrow t'$, i.e.
\begin{align} \label{eq:symK}
    K_{t',t,\sigma}(x,y) = K_{t,t',\sigma}(y,x).
\end{align}
The $\ell^1$ condition (\ref{cond:l1}) for the function $\sigma$ implies the following
\begin{prop} \label{prop:prop1}
The restriction of $K_{t,t',\sigma}$ on $\ell^2 \{n +1/2,n+3/2,\dots \}$ is trace class for any $n \in \Z$.
\end{prop}
The above proposition allows us to consider, for $n \in \Z$, the Fredholm determinant
\begin{align} \label{def:tau-sigma}
\tau_n (t,t';\sigma) := Z_{t,t'} \det(1 - K_{t,t',\sigma})_{\ell^2 \{n+1/2, n+3/2, \dots\}}, 
\end{align}
which is expanded as
\begin{align}
    \det(1 - K_{t,t',\sigma})_{\ell^2 \{n+1/2, n+3/2, \dots\}} = \sum_{m \geq 0} (-1)^m \sum_{\{x_1,\dots,x_m \} \subset \Z'} \det(K_{t,t',\sigma}(x_i,x_j))_{i,j=1}^m.
\end{align}
When the function $\sigma$ is the indicator function of the positive half integer $\mathfrak{1}_{\Z'_{>0}}$, we write simply 
\begin{align} \label{eq:defKshort}
   K_{t,t'} := K_{t,t', \mathfrak{1}_{\Z'_{>0}}}. 
\end{align}
We recall in Section \ref{section:schur} below that the kernel $K_{t,t'}$ is the correlation kernel of the determinantal point process given by the image of the Schur measure with parameters $t$ and $t'$ on the space of configurations on $\Z'$. As announced, the Fredholm determinant (\ref{def:tau-sigma}) is the expectation of a multiplicative functional for the latter Schur measure, see Section \ref{section:schur} for definitions and notation.
\begin{prop} \label{prop:prop2} Let $\P_{t,t'}$ be the Schur measure with parameters $t,t'$.
Then, we have
\begin{align*}
\tau_n(t,t';\sigma) = Z_{t,t'} \E_{\P_{t,t'}} \left[ \prod_{x \in \mathfrak{S}_0(\lambda)} (1- \sigma(x-n) )\right].
\end{align*}
\end{prop}
Our main result is the following.
\begin{thm} \label{thm:thm1} The functions $\tau_n(t,t';\sigma)$ are $\tau$-functions of the 2D Toda lattice hierarchy, in the sense that they satisfy bilinear Hirota equations:
\begin{multline} \label{eq:hirota}
[z^{l-m}]\gamma(z^{-1},-2s')\tau_{m+1}(t+s,t'+s'+\{z\};\sigma) \tau_l(t-s,t'-s'-\{z\} ; \sigma ) \\
=[z^{m-l}] \gamma(z^{-1},2s) \tau_m (t+s -\{z\} , t +s' ; \sigma) \tau_{l+1}(t-s+ \{z \}, t-s'; \sigma ),
\end{multline}
where $s=(s_1,s_2,\dots)$, $s'=(s_1' , s_2' , \dots ) \subset \C$ are sequences satisfying the same convergence condition as $t$ and $t'$, where
\begin{align*}
\{z\} = \left( z, \frac{z^2}{2}, \frac{z^3}{3}, \dots \right),
\end{align*}
and where $[z^k] F(z)$ is the coefficient in $z^k$ in the Laurent series expansion of $F(z)$.
\end{thm}
Observe that (\ref{eq:hirota}) is a partial differential equation in the parameters $t$ and $t'$, namely the one from \cite{uenotakasaki}, Theorem 1.11. Indeed, Taylor formula implies that, for an analytic function $f(x)=f(x_1,x_2,\dots)$, we have
\begin{align*}
f(x+ \{z\}) = \sum_{ n \geq 0} \mathrm{p}_n \left( \partial_{x_1} , \frac{\partial_{x_2}}{2}, \dots \right) (f) z^n,
\end{align*}
where the polynomial $\mathrm{p}(a_1,a_2,\dots)$ is defined by
\[\exp \left( \sum_{k \geq 1} a_k z^k \right) = \sum_{n \geq 0} \mathrm{p}_n(a_1,a_2,\dots) z^n. \]

\begin{rem}We decided here to adopt an analytic rather than purely formal point of view. The convergence conditions for the sets of parameters $t$ and $t'$ could certainly be weakened, as they are not strictly necessary for the Schur measures to exist and to have determinantal expressions for their correlation functions. A necessary condition for the Schur measures to be well defined is that the sum
\begin{align*}
    \sum_{n \geq 1} nt_n t'_n,
\end{align*}
is convergent, which is achieved if for instance the series $\sum_{n \geq 1}t_nz^n + t_n'z^{-n}$ belongs to the Sobolev space $H^{1/2}(\T)$. However, it is not clear to us whether the latter condition is sufficient for the determinantal identities we need to hold, in particular for the relevant operators to be Hilbert-Schmidt, see Sec. \ref{sec:proofofprops} and Prop. \ref{prop:hilbertschmidt}; a sufficient condition would be that this series belongs to $L^\infty(\T)$. So, instead of searching for the most general conditions for the sets of parameters $t$ and $t'$, we decided to use a quite strong convergence condition, which makes all the arguments concerning analytic quantities clear. In particular, we can establish that the tau-functions are analytic in the time parameters $t_1,\dots,t_1',\dots$ if they are sufficiently small, see Remark \ref{rem:tauanalytic} below.
\end{rem}
\subsection{Proof's strategy and organisation of the paper}
The proof of Theorem \ref{thm:thm1} first lies on the fact that the function $\tau_n(t,t';\sigma)$ can be seen as the average of a multiplicative functional  under a Schur measure, see Theorem \ref{thm:okounkov} and Proposition \ref{prop:prop2}. The argument is then based on the semi-infite wedge space formalism of the fermionic Fock space from which we can construct tau-functions for the 2D Toda lattice hierarchy, and we show that multiplicative functionals of Schur measures fit into that formalism.\\

We recall the necessary material on Schur measures in Section \ref{section:schur}, where we also describe their finite temperature generalizations and view them as determinantal point processes with correlation kernel $K_{t,t',\sigma}$ for a particular choice of the function $\sigma$, see Theorem \ref{thm:borodin}.\\

Section \ref{sec:proofofprops} is devoted to the proofs of Propositions \ref{prop:prop1} and \ref{prop:prop2}, where we also show that the function $\tau_n(t,t';\sigma)$ is analytic in each of the variables $t_k,t_k'$ near the origin.\\

We further recall the semi-infinite wedge formalism in Section \ref{section:infinitewedge} and express the correlation functions for the Schur measures as vacuum expectation values of some specific operators on the fermionic Fock space, see Proposition \ref{prop:schurpsi}. By means of the Boson-Fermion correspondence, we then see in Proposition \ref{prop:schurtau} how to generate solutions of the bilinear Hirota equations (\ref{eq:hirota}) from expectation values of a particular abstract operator on the fermionic Fock space satisfying a commutation condition.\\

We conclude the proof of Theorem \ref{thm:thm1} in Section \ref{section:proof} by showing that the functions $ \tau_n(t,t';\sigma)$ fit into the formalism of Propositions \ref{prop:schurpsi} and \ref{prop:schurtau}, see Lemmas \ref{lem:commutation} and \ref{lem:expansion}.\\

\subsection*{Acknowledgement} I am deeply grateful to Mattia Cafasso, Giulio Ruzza, and Tamara Grava for enlightening discussions.

This project received funding from project ULIS 2023-09915 from R\'egion Pays de la Loire, and of the fellowship "Assegni di ricerca FSE SISSA 2019"  from Fondo Sociale Europeo - Progetto SISSA OPERAZIONE 1 codice FP195673001.

The author acknowledges support by FNRS Research Project T.0028.23 and by the Fonds Sp\'ecial de Recherche of UCLouvain.

\section{On Schur measures and finite temperature Schur measures} \label{section:schur}
\subsection{Schur measures}
A partition, or equivalently a Young diagram, is a non-increasing sequence $\lambda= (\lambda_1 \geq \lambda_2 \geq \dots)$ of non-negative integers that is eventually zero. The length of a Young diagram is the index of its last non-zero entry and is denoted by $l (\lambda)$. The set of all Young diagrams is denoted by $\Y$. To a Young diagram and an integer $n \in \Z$, we associate a subset of $\Z'$ by
\begin{align}
\lambda \mapsto \mathfrak{S}_n(\lambda) := \{ \lambda_i -i +1/2 + n \}.
\end{align}

We let $\Lambda$ be the algebra over $\C$ of symmetric functions, see e.g. the book \cite{macdonald} for the definitions and for the proofs of the identities used in this section. It is spanned by the Newton power sums defined by
\begin{align*}
p_k( \mathrm{x}_1, \mathrm{x}_2, \dots ) = \sum_{ i =1}^{ + \infty} \mathrm{x}_i^k, \quad k=1,2,\dots
\end{align*}
The Schur functions $s_\lambda$, $\lambda \in \Y$, are defined by the Jacobi-Trudi Formula
\begin{align*}
s_\lambda := \det( h_{\lambda_i -i + j} )_{i,j=1}^N, \quad N \geq l(\lambda),
\end{align*}
where $h_k$ is the complete homogeneous function:
\begin{align*}
h_k( \mathrm{x}_1,\mathrm{x}_2,\dots) = \sum_{i_1 \leq i_2 \leq \dots \leq i_k} \mathrm{x}_{i_1} \cdots x_{i_k}, \quad k \in \Z
\end{align*}
with the convention that $h_0 =1$ and $h_k=0$ for $k<0$.\\

For $\lambda \in \Y$ and $t=t_1,t_2,\dots \subset \C$, we write $s_\lambda(t)$ for the image of the Schur function $s_\lambda$ given by the algebra morphism from $\Lambda$ to $\C$ defined by
\begin{align*}
p_k \mapsto \frac{1}{k}t_k, \quad k=1,2, \dots
\end{align*}
At the formal level, we have the Cauchy identity
\begin{align*}
    \sum_{ \lambda \in \Y} s_\lambda(\mathrm{x}_1,\mathrm{x}_2,\dots) s_\lambda(\mathrm{y}_1,\mathrm{y}_2,\dots) = \prod_{i,j=1}^{\infty} \frac{1}{1-\mathrm{x}_i\mathrm{y}_j},
\end{align*}
which, under suitable convergence conditions, translates into
\begin{align*}
    \sum_{\lambda \in \Y} s_\lambda(t)s_\lambda(t')=Z_{t,t'},
\end{align*}
where $Z_{t,t'}$ is defined in \eqref{eq:defZ}.
\begin{defi}The Schur measure with parameters $t, t'$ is the signed probability measure $\P_{t,t'}$ on $\Y$ given by
\begin{align*}
\P_{t,t'}(\lambda) := Z_{t,t'}^{-1} s_\lambda(t) s_\lambda(t').
\end{align*}
\end{defi}
\subsection{Finite temperature Schur measures}
The Schur measures may be generalized as follows, see \cite{borodinperiodic}, \cite{beteabouttier}. For Young diagrams $\mu, \lambda \in \Y$, we write $\mu \subset \lambda$ if $\mu_i \leq \lambda_i$ for all $i=1,2 \dots$, and denote by $\lambda / \mu$ the sequence of non-negative integers $(\lambda_1 - \mu_1, \lambda_2 - \mu_2,\dots)$. The skew Schur function indexed by the skew diagram $\lambda / \mu$ is defined by
\begin{align*}
s_{\lambda / \mu} = \det ( h_{ \lambda_i - i - \mu_j + j} )_{i,j =1}^N, \quad N \geq l(\lambda).
\end{align*}
\begin{defi}For $u \in [0,1)$, and sets of parameters $t,t' \subset \C$, the \emph{finite temperature Schur measure} $\P_{u,t,t'}$  with parameters $u,t,t'$ is the signed probability measure on $\Y$ defined as
\begin{align*}
\P_{u,t,t'}(\lambda)= Z_{u,t,t'}^{-1} \sum_{  \mu \subset \lambda} u^{|\mu|} s_{\lambda/ \mu}(t) s_{\lambda /  \mu} (t') , \quad \lambda \in \Y.,
\end{align*}
where $Z_{u,t,t'}$ is the normalizing constant given by
\begin{align*}
    Z_{t,t',u}=\prod_{k \geq 1} \frac{1}{1-u^k}Z_{u^k/(1-u^k)t,t'.}
\end{align*}
\end{defi}
Observe that for $u=0$, one recovers the definition of Schur measures.
\subsection{Schur measures and finite temperature Schur measures as determinantal point processes}
One of the main results of Okounkov in \cite{okounkovschurmeasures} is the following theorem.
\begin{thm} \label{thm:okounkov} For any finite set $X= \{x_1, \dots, x_m\} \subset \Z'$, we have
\begin{align*}
\P_{t,t'} ( X \subset \mathfrak{S}_0( \lambda ) ) = \det \left( K_{t,t'} (x_i, x_j) \right)_{i,j=1}^m,   
\end{align*}
where the correlation kernel $K_{t,t'}$  is given by (\ref{eq:defK}) and (\ref{eq:defKshort}):
\begin{align*}
K_{t,t'} (x,y) = \sum_{ k \in \Z', \hspace{0.1cm} k>0} J_{x+k}(t,t') J_{-y-k}(-t,-t').
\end{align*}
\end{thm}
The above result has been generalized by Borodin in \cite{borodinperiodic} (see also \cite{beteabouttier}) to finite temperature Schur measures as follows.
\begin{thm} \label{thm:borodin} Let $u \in [0,1)$ and consider the parameters
\begin{align*}
t = \frac{1}{1-u}( \tilde{t}_1, \tilde{t}_2, \dots) , \quad t' =  \frac{1}{1-u}(\tilde{t}'_1, \tilde{t}'_2, \dots ),
\end{align*}
for fixed $\tilde{t}_k, \tilde{t}'_k$, $k=1,2,\dots$. Let $\P_{u, \tilde{t}, \tilde{t}'}$ be the finite temperature Schur measure with parameters $u, \tilde{t}, \tilde{t}'$. Let $c$ be a $\Z$-valued random variable defined on some probability space $(\Omega, P)$ such that
\begin{align*}
P(c=n) =  \frac{u^{ \frac{n^2}{2}}}{\theta_3(1, u)},
\end{align*}
where 
\begin{align*}
\theta_3(1,u)= \sum_{ n \in \Z} u^{\frac{n^2}{2}}.
\end{align*}
Then we have for any finite set $X = \{x_1, \dots x_m \} \subset \Z'$,
\begin{align*}
\P_{u,\tilde{t}, \tilde{t}'} \otimes P ( X \subset \mathfrak{S}_c (\lambda) ) = \det \left( K_{t,t',\sigma_u} (x_i, x_j) \right)_{i,j=1}^m,
\end{align*}
where $K_{t,t',\sigma_u}$ is given by (\ref{eq:defK}) with
\begin{align*}
\sigma_u(k) = \frac{1}{1-u^k}, \quad k \in \Z'.
\end{align*}
\end{thm}
\section{Proof of Propositions \ref{prop:prop1} and \ref{prop:prop2}} \label{sec:proofofprops}
For both the proofs of Prop. \ref{prop:prop1} and \ref{prop:prop2}, we note that the operator $K_{t,t',\sigma}$ can be factorized as:
\begin{align} \label{eq:factorKsigma}
    K_{t,t',\sigma} = H_{t,t'}\sigma \check{H}_{t,t'},
\end{align}
where $H_{t,t'}$ and $\check{H}_{t,t'}$ are the Hankel operators acting on a function $f \in \ell^2(\Z')$ by
\begin{align*}
H_{t,t'} f (x) &= \sum_{ k \in \Z'} J_{k+x}(t,t') f(k), \\
\check{H}_{t,t'} f(x) &= \sum_{k \in \Z'} J_{-k - x} (-t,-t') f(k) \\
&=\sum_{k \in \Z'}J_{k+x}(t',t)f(k)= H_{t',t}f(x),
\end{align*}
where the last line follows from (\ref{eq:J3}). In order to prove Proposition \ref{prop:prop1}, it thus suffices to establish the following
\begin{prop} \label{prop:hilbertschmidt}
    The operators $\mathfrak{1}_{\{n+1/2,n+3/2,\dots\}}H_{t,t'}\sqrt{\sigma}$ and $\sqrt{\sigma}\check{H}_{t,t'} \mathfrak{1}_{\{n+1/2, n + 3/2, \dots\}}$ are Hilbert-Schmidt operators on $\ell^2(\Z')$.
\end{prop}
\begin{proof}
By symmetry, we just prove that the operator $\sqrt{\sigma}\check{H}_{t,t'} \mathfrak{1}_{\{n+1/2, n + 3/2, \dots\}}$ is Hilbert-Schmidt, i.e., that the sum
\begin{align*}
    \sum_{x \in \Z', \, y\geq n+1/2} \sigma(x) |J_{x+y}(t',t)|^2
\end{align*}
is finite. We will treat the sum over positive and negative $x \in \Z'$ separately. By the Cauchy formula, we have
\begin{align*}
    |J_{x+y}(t',t)|^2 &=  \int_\T \frac{dw}{2i \pi}\int_\T \frac{dz}{-2i \pi} J(z;t',t)\overline{J(w;t',t)} \left(\frac{w}{z}\right)^{x+y+1} \\
    &=\int_\T \frac{dw}{2i \pi}\int_\T \frac{dz}{-2i \pi} J(z;t',t)J(w;\overline{t},\overline{t'}) \left(\frac{w}{z}\right)^{x+y+1}.
\end{align*}
Since the function $J$ is assumed to be analytic in a neighborhood of the unit circle, we can deform the contour of integration for $z$ into the circle of radius $1+\epsilon$, denoted below by $\T_\epsilon$, for $\epsilon>0$ small enough. Since we now have $|z| > |w|$ for all $z \in \T_\epsilon$, $w \in \T$, we can sum over $y \geq n+1/2$ and use Fubini to obtain
\begin{align*}
    \sum_{ y \geq n+1/2} |J_{x+y}(t',t)|^2 =  \int_\T \frac{dw}{2 i\pi}  \int_{\T_\epsilon}\frac{dz}{-2i \pi}  \, J(z;t',t) J(w;\overline{t},\overline{t'})\left(\frac{w}{z}\right)^{n+x+3/2} \frac{z}{z-w}.
\end{align*}
We deduce that there exists $C>0$ such that for all $x \in \Z'$, we have
\begin{align*}
    \sum_{ y \geq n+1/2} |J_{x+y}(t',t)|^2  \leq \frac{C}{(1 + \epsilon)^x},
\end{align*}
and since the function $\sigma$ is bounded, we obtain that
\begin{align*}
     \sum_{x \in \Z'_{>0}, \, y\geq n+1/2} \sigma(x) |J_{x+y}(t',t)|^2 < +\infty.
\end{align*}
We treat the sum over negative $x$ as follows. For any $w \in \T$, by a residue computation, we have
\begin{align*}
     \int_{\T_\epsilon} \frac{dz}{2i \pi} \, J(z;t',t) \frac{z}{z^{n+x+3/2}(z-w)} = \int_{\T_\epsilon} \frac{dz}{2 i \pi} \,  J(z;t',t) \frac{1}{z^{n+x+1/2}} +  \frac{J(w;t',t)}{w^{n+x+1/2}},
\end{align*}
from which it follows that
\begin{align*}
    \sum_{ y \geq n+1/2} |J_{x+y}(t',t)|^2 &=  -\int_\T \frac{dw}{2 i \pi} \int_{\T_\epsilon} \frac{dz} {2 i \pi} \,J(z;t',t) J(w;\overline{t},\overline{t'}) \left(\frac{w}{z}\right)^{n+x+1/2}w \\
    &- \int_\T \frac{dw}{2 i \pi} \, J(w,t',t)J(w, \overline{t}, \overline{t'})w.
\end{align*}
The second term is independent of $x$, while we can again deform the contour of integration in the first term to obtain a quantity that is bounded independently of $x$. Since the function $\sigma$ is summable over the negative half-integers, we obtain the desired statement.
\end{proof}
\begin{proof}[Proof of Proposition \ref{prop:prop2}]
By Theorem \ref{thm:okounkov}, we have
\begin{align*}
    \E_{\P_{t,t'}} \left[ \prod_{x\in \mathfrak{S}_0(\lambda)}(1-\sigma(x-n) )\right] = \det(1- \sigma(\cdot-n)K_{t,t'})_{\ell^2(\Z')},
\end{align*}
so it suffices to prove that
\begin{align*}
    \det(1-K_{t,t',\sigma})_{\ell^2\{n+1/2,n+3/2,\dots\}}= \det(1- \sigma(\cdot-n)K_{t,t'})_{\ell^2(\Z')}.
\end{align*}
Using the factorization (\ref{eq:factorKsigma}), we use of the identity
\begin{align*}
\det(1 + AB ) = \det(1 +BA),
\end{align*}
valid when $A$ and $B$ are Hilbert-Schmidt operators. We obtain
\begin{align*}
\det \left( 1 -K_{t,t'\sigma} \right)_{\ell^2 \{n+1/2,n+3/2, \dots \}} &= \det \left( 1 - \mathfrak{1}_{ \{n+1/2,n+3/2, \dots \} } H_{t,t'} \sigma \check{H}_{t,t'} \mathfrak{1}_{ \{n+1/2,n+3/2,\dots \} } \right)_{\ell^2(\Z')}\\
&= \det \left(1 - \sqrt{\sigma} \check{H}_{t,t'} \mathfrak{1}_{\{n+1/2,n+3/2,\dots\}}H_{t,t'} \sqrt{\sigma} \right)_{\ell^2(\Z')}.
\end{align*}
We can now observe that the kernel of the left factor can be expressed as:
\begin{align*}
    \left(\sqrt{\sigma} \check{H}_{t,t'} \mathfrak{1}_{\{n+1/2,n+3/2\dots\}}\right)(x,y)&=\sqrt{\sigma(x)}J_{x+y}(t',t)\mathfrak{1}_{\{n+1/2,\dots\}}(y) \\
    &= \sqrt{\sigma(x)}J_{x+y}(t',t)\mathfrak{1}_{\Z'_{>0}}(y-n) \\
    &=\sqrt{\sigma(x+n-n)} J_{x+n+y-n}(t',t) \mathfrak{1}_{\Z_{>0}}(y-n)\\
    &=\left(T^n\sqrt{\sigma(\cdot-n)}\check{H}_{t,t'}\mathfrak{1}_{\Z_{>0}}\right)(x,y-n),
\end{align*}
where $T: \ell^2(\Z') \to \ell^2(\Z')$ is the shift operator: $Tf(x)=f(x+1), \, f \in \ell^2(\Z')$,  $x \in \Z'$. Similarly, the kernel of the right factor can be written as
\begin{align*}
    \left( \mathfrak{1}_{\{n+1/2,n+3/2,\dots\}} H_{t,t'}\sqrt{\sigma} \right)(y,x') = \left(\mathfrak{1}_{\Z'_{>0}} H_{t,t'}\sqrt{\sigma(\cdot-n)}T^{-n}\right)(y-n,x').
\end{align*}
It follows that
\begin{align*}
    \left(\sqrt{\sigma} \check{H}_{t,t'} \mathfrak{1}_{\{n+1/2,n+3/2,\dots\}}H_{t,t'} \sqrt{\sigma}\right)(x,x')=\left(T^n\sqrt{\sigma(\cdot -n)}\check{H}_{t,t'}\mathfrak{1}_{\Z'_{>0}}H_{t,t'}\sqrt{\sigma(\cdot -n)} T^{-n}\right)(x,x'),
\end{align*}
and thus
\begin{align*}
\det \left( 1 -K_{t,t'\sigma} \right)_{\ell^2 \{n+1/2,n+3/2, \dots \}}&= \det \left( 1- \sqrt{\sigma (\cdot -n) } \check{H}_{t,t'} \mathfrak{1}_{\Z'_{> 0}} H_{t,t'} \sqrt{\sigma( \cdot - n) } \right)_{\ell^2(\Z')} \\
&= \det \left( 1- \sigma( \cdot - n ) \check{ H}_{t,t'} \mathfrak{1}_{ \Z'_{> 0}} H_{t,t'}  \right)_{\ell^2(\Z')}\\
&= \det\left(1- \sigma(\cdot-n)K_{t',t}\right)_{\ell^2(\Z')}\\
&=\det\left(1- \sigma(\cdot -n)K_{t,t'}\right)_{\ell^2(\Z')},
\end{align*}
where the last equality follows from the symmetry $K_{t',t}(x,y)=K_{t,t'}(y,x)$, which does not affect the value of the determinant. The Proposition is proved.
\end{proof}
\begin{rem} \label{rem:tauanalytic}
From the same computations as in the proof of Prop. \ref{prop:prop1}, first using the Cauchy formula and then deforming the contours of integration, we can write
\begin{align*}
    K_{t,t'}(x,y)= \int_{\T} \frac{dw}{2i \pi} \int_{\T_{\epsilon}} \frac{dz}{2i \pi} \frac{J(z;t,t')}{J(w;t,t')} \frac{1}{(z-w)z^{x+1/2}w^{-y+1/2}}
\end{align*}
Then, by similar arguments as in the proof of Prop. \ref{prop:prop1}, we see that uniformly in each $t_k$ and each $t'_k$ in a compact set near the origin, $K_{t,t'}(x,y)$ decays exponentially when $x \neq y$ or $x=y>0$ and remains bounded when $x=y<0$. We deduce that $\tau_n(t,t';\sigma)$ is an analytic function in each of the variables $t_k$, $t_k'$.
\end{rem}
\section{The semi-infinite wedge formalism} \label{section:infinitewedge}
\subsection{The fermionic Fock space}
 Let $\Lambda(\mathbb{Z}')$ be the set of all subsets $S \subset \mathbb{Z}'$ such that $|S^+|:= |\mathbb{Z}'_{>0} \cap S| < +\infty$ and $|S^-|:=|\mathbb{Z}'_{<0} \setminus \left(S \cap \mathbb{Z}'_{<0} \right)| < +\infty$. The fermionic Fock space is by definition the Hilbert space freely spanned by $\Lambda(\mathbb{Z}')$, which we denote by $\Lambda^{\frac{\infty}{2}}(\Z')$. For $S \in \Lambda(\mathbb{Z}')$, we denote by $v_S$ the vector indexed by $S$ and write
\begin{align*}
v_S:= \underline{s_1} \wedge \underline{s_2} \wedge ...
\end{align*}
where $S= \lbrace s_1 > s_2 > ... \rbrace$. The Hilbert space structure on the fermionic Fock space $\Lambda^{\frac{\infty}{2}}(\Z')$ comes from the inner product $\langle.,.\rangle $ such that $v_S$, $S \in \Lambda(\mathbb{Z}')$, is an orthonormal basis.\\

For an integer $n \in \mathbb{Z}$, let $\Lambda_n(\mathbb{Z}')$ be the set of all $S \in \Lambda(\mathbb{Z}')$ such that $|S^+|-|S^-|=n$. We have
\begin{align*}
\Lambda(\mathbb{Z}')= \sqcup_{n \in \mathbb{Z}}\Lambda_n(\mathbb{Z}'),
\end{align*}
and the canonical inclusion $\iota : \Lambda(\mathbb{Z}') \hookrightarrow \Lambda^{\frac{\infty}{2}}(\Z')$ gives rise to a direct sum decomposition of $\Lambda^{\frac{\infty}{2}} (\Z')$, called the charge decomposition:
\begin{align*}
\left.\Lambda \right.^{\frac{\infty}{2}}(\Z')=\oplus_{n \in \mathbb{Z}}\left.\Lambda\right._n^{\frac{\infty}{2}}(\Z')
\end{align*}
where $\Lambda_n^{\frac{\infty}{2}}(\Z')= \iota (\Lambda_n(\mathbb{Z}'))$.\\

Recall that, for each $n \in \Z$, one embeds $\Y$ into $\Lambda_n(\Z')$ by the map
\begin{align*}
\lambda \mapsto \mathfrak{S}_n(\lambda) := \{ \lambda_i -i + 1/2 + n, \hspace{0.1cm} i =1,2,\dots \}.
\end{align*}
We write $v_\lambda:=v_{\mathfrak{S}_0(\lambda)}$. We also use the following notation: for $n \in \mathbb{Z}$, set
\begin{align*}
v_n:=\underline{n-1/2} \wedge
 \underline{n-3/2} \wedge ... \in \Lambda^{\frac{\infty}{2}}_n (\Z').
\end{align*}
Observe that $v_0=v_{\emptyset}$, where $\emptyset$ is the empty Young diagram.

\subsection{Operators on $\Lambda^{\frac{\infty}{2}}(\Z')$}
For $k \in \Z'$, define the creation operator $\psi_k$ as being the operator of exterior multiplication by $\underline{k}$: if $S=\{ s_1 > s_2 > ... \} $,
\begin{equation} \label{eq:defpsi}
\psi_k v_S = \underline{k} \wedge \underline{s_1} \wedge \underline{s_2} \wedge ... \\
 =  \begin{cases} 0 &
\text{ if $k \in S$} \\
(-1)^i \underline{s_1} \wedge \cdots \wedge \underline{s_{i}} \wedge \underline{k} \wedge \underline{s_{i+1}}\wedge \cdots &
\text{ if $s_i > k > s_{i+1}$.}
\end{cases} 
\end{equation}
For $k \in \mathbb{Z}'$, define the annihilation operator $\psi_k^*$ by:
\begin{equation} \label{eq:defpsi*}
\psi_k^*v_S =\begin{cases} 0 &
\text{ if $k \notin S$} \\
(-1)^{i-1}\underline{s_1}\wedge \cdots \wedge \underline{s_{i-1}} \wedge \underline{s_{i+1}}\wedge \cdots &
\text{ if $k=s_i$}.
\end{cases}
\end{equation}
The operator $\psi_k^*$ is the adjoint operator of $\psi_k$, and we have
\begin{equation*}
\psi_k \psi_k^* v_S = \begin{cases} v_S & \text{if $k \in S$} \\
0 & \text{if $k \notin S$},
\end{cases}
\end{equation*}
and
\begin{equation*}
\psi_k^* \psi_k v_S = \begin{cases} v_S & \text{if $k \notin S$} \\
0 & \text{if $k \in S$}
\end{cases}
\end{equation*}
We have the following anti-commutation relations:
\begin{align*}
\psi_k \psi_l^* + \psi_l^*\psi_k = \delta_{kl}, \quad
\psi_k \psi_l + \psi_l \psi_k=0, \quad
\psi_k^*\psi_l^*+\psi_l^*\psi_k^* &=0.
\end{align*}
In particular, for any $\{k_1,\dots, k_m \} \subset \Z'$, the product $ \psi_{k_1}\psi_{k_1}^*\cdots \psi_{k_m}\psi_{k_m}^*$ does not depend on the order of the $k_i$'s. \\

The charge operator $C$ can be defined by
\begin{align*}
C:= \sum_{ k \in \Z'_{>0}} \psi_k \psi_k^* - \psi_{-k}^* \psi_{-k},
\end{align*}
and we have
\begin{align*}
\ker (C- nI) = \Lambda_n^{\frac{\infty}{2}}(\Z').
\end{align*}

The shift operator is defined by
\begin{align*}
R v_S = \underline{s_1 +1} \wedge \underline{s_2 +1} \wedge \dots, \quad S= \{s_1> s_2> \dots \}  \in \Lambda(\Z').
\end{align*}

We introduce the bosonic operators on $\Lambda^{\frac{\infty}{2}}(\Z')$ as follows: for $n \in \mathbb{Z}\setminus \lbrace 0 \rbrace$, we define the operator $\alpha_n$ by:
\begin{align*}
\alpha_n := \sum_{k \in \mathbb{Z}'} \psi_{k-n}\psi_k^*
\end{align*}

We introduce the generating series
\begin{align*}
\psi(z) := \sum_{i \in \mathbb{Z}'} \psi_i z^i, \quad
\psi^*(w) := \sum_{j \in \mathbb{Z}'} \psi_j^*w^{-j}.
\end{align*}
We have the following commutation relations
\begin{align*}
[\alpha_n ,\alpha_m] = n \delta_{n,-m}, \quad
[\alpha_n , \psi(z) ] = z^n \psi(z), \quad
[\alpha_n , \psi^*(w) ] = -w^n \psi^*(w).
\end{align*}

The vertex operators $\Gamma_{\pm}(t)$ are defined by
\begin{align*}
\Gamma_\pm(t)=\exp\left(\sum_{n \geq 1} t_n \alpha_{\pm n}\right).
\end{align*}
They satisfy
\begin{align} \label{eq:vertexcharge}
\Gamma_+(t) v_m = v_m
\end{align}
and the commutation relations
\begin{align} \label{eq:commutevertex}
\Gamma_+(t) \Gamma_-(t') &= Z_{t,t'} \Gamma_-(t') \Gamma_+(t) \\ \label{eq:commutevertex1}
\Gamma_{\pm }(t) \psi(z) &= \gamma ( z^{ \pm 1} , t ) \psi(z) \Gamma_{ \pm}(t) \\ \label{eq:commutevertex2}
\Gamma_{ \pm}(t) \psi^*(z) &= \gamma( z^{ \pm 1} , t)^{-1} \psi^*(z) \Gamma_{ \pm}(t). 
\end{align}
We have the Boson-Fermion correspondence allowing to recover the operators $\psi_k$ from the operators $\alpha_n$:
\begin{align} \label{eq:bosonfermioncorres}
\psi(z) = z^C R \Gamma_-( \{z\} ) \Gamma_+( - \{z^{-1}\} ), \quad
\psi^*(z) = R^{-1} z^{-C} \Gamma_-( -\{z \} ) \Gamma_+( \{z^{-1} \} ).
\end{align}

\subsection{Intermediate results}

We recall an intermediate result from Okounkov on Schur measures, expressing the correlation functions in terms of vacuum expectation values for some operators on the fermionic Fock space, and that is used to prove Thm. \ref{thm:okounkov} above. It is based on the fact that Schur functions themselves are expressed as
\begin{align*}
    s_\lambda(t)=\langle \Gamma_-(t) v_\emptyset,v_\lambda\rangle,
\end{align*}
which follows from the commutation relations \eqref{eq:commutevertex1}, \eqref{eq:commutevertex2}.
\begin{prop} \label{prop:schurpsi} Let $\P_{t,t'}$ be the Schur measure with parameters $t,t'$. Then we have for any finite set $X=\{x_1,\dots, x_m \} \subset \Z'$
\begin{align*}
\P_{t,t'} ( X \subset \mathfrak{S}_0(\lambda) )= Z_{t,t}^{-1} \langle \Gamma_+(t) \psi_{x_1} \psi_{x_1}^* \cdots \psi_{x_m} \psi_{x_m}^* \Gamma_-(t') v_{\emptyset} , v_{\emptyset} \rangle .
\end{align*}
\end{prop}
We also recall the following, that is the stepping stone in the proof of our main result. Its proof lies on the Boson-Fermion correspondence.
\begin{prop} \label{prop:schurtau} Let $A$ be an operator on $\Lambda^{ \frac{\infty}{2}}$ such that $A \otimes A$ commutes with the operator
\begin{align} \label{eq:defPsi}
\Psi := \sum_{ k \in \Z'} \psi_k \otimes \psi_k^*.
\end{align}
Set
\begin{align*}
\tilde{A}:= \Gamma_+(t) A \Gamma_-(t').
\end{align*}
Then, the sequence formed by the functions
\begin{align*}
 \tilde{\tau}_n(t,t'; A):=\langle \tilde{A} v_n, v_n \rangle
\end{align*}
are $\tau$-functions for the Toda lattice hierarchy, in the sense that they satisfy the bilinear Hirota equations
\begin{multline} \label{eq:hirotagene}
[z^{l-m}]\gamma(z^{-1},-2s')\tilde{\tau}_{m+1}(t+s,t'+s'+\{z\};A) \tilde{\tau}_l(t-s,t'-s'-\{z\} ; A ) \\
=[z^{m-l}] \gamma(z^{-1},2s) \tilde{\tau}_m (t+s -\{z\} , t +s' ; A) \tilde{\tau}_{l+1}(t-s+ \{z \}, t-s'; A )
\end{multline}
\end{prop}
Our goal in the next section will be to construct an operator $A=A_\sigma$ for which the corresponding tau-function given by the preceding proposition will be the function $\tau_n(t,t';\sigma)$ (\ref{def:tau-sigma}), or, equivalently by Prop. \ref{prop:prop2}, multiplicative statistics of Schur measures.
\begin{proof}[Proof of Prop. \ref{prop:schurtau}]
We sketch a proof for completeness. Since the operator $A \otimes A$ commutes with $\Psi$, so does the operator $\tilde{A}$. This relation translates into
\begin{align*}
[ z^0 ] \left( \tilde{A} \otimes \tilde{A} \right) \left( \psi(z) \otimes \psi^*(z) \right) = [z^0] \left( \psi(z) \otimes \psi^*(z) \right) \left( \tilde{A} \otimes \tilde{A} \right),
\end{align*}
i.e.
\begin{align} \label{eq:commute}
[z^0] \Gamma_+(t)A\Gamma_-(t') \psi(z) \otimes \Gamma_+(t)A\Gamma_-(t')  \psi^*(z) = [z^0] \psi(z) \Gamma_+(t)A\Gamma_-(t')  \otimes \psi^*(z) \Gamma_+(t)A\Gamma_-(t').
\end{align}
Relation (\ref{eq:hirotagene}) now follows by taking the matrix coefficient of eq. (\ref{eq:commute}) for the vector \[\Gamma_-(s') v_m \otimes \Gamma_-(-s') v_{l+1}\]
against
\[ \Gamma_-(s) v_{m+1} \otimes \Gamma_-(-s) v_{l}. \]
Indeed, from the Boson-Fermion correspondence \eqref{eq:bosonfermioncorres}, we have for $m \in \Z$ and a sequence $s'$
\begin{align*}
\Gamma_+(t) A \Gamma_-(t')\psi(z) \Gamma_-(s') v_m &= \Gamma_+(t) A \Gamma_-(t')z ^{C} R \Gamma_-( \{ z \} ) \Gamma_+( - \{ z^{-1} \} ) \Gamma_-( s') v_m \\
&= z^{m+1} \gamma ( z^{-1} , -s') \Gamma_+(t) A \Gamma_-(t'  + \{z\} + s' )v_{m+1}.
\end{align*}
The second equality above is obtained as follows: from (\ref{eq:commutevertex}), we have
\[ \Gamma_+( - \{ z^{-1} \} ) \Gamma_-( s') = \gamma ( z^{-1} , -s') \Gamma_-( s') \Gamma_+( - \{ z^{-1} \} ). \]
Then, we use (\ref{eq:vertexcharge}), the fact that $R$ commutes with the vertex operators and changes $v_m$ to $v_{m+1}$, and finally that the vertex operators preserve the charge, producing the factor $z^{m+1}$. 

Similarly, we obtain for $l \in \Z$:
\begin{align*}
\Gamma_+(t)A\Gamma_-(t')  \psi^*(z) \Gamma_-(-s') v_{l+1} =  z^{-l-1} \gamma ( z^{-1} , -s') \Gamma_+(t)A \Gamma_-(t' -s' - \{ z \}) v_l.
\end{align*}
We deduce that
\begin{multline*}
[z^0]\left( \tilde{A} \psi(z) \otimes \tilde{A} \psi^*(z) \right) \left( \Gamma_-(s') v_m \otimes \Gamma_-(-s') v_{l+1} \right) \\
= [z^{l-m}] \gamma( z^{-1} , -2s' ) \Gamma_+(t) A \Gamma_-(t'  + \{z\} + s' )v_{m+1} \otimes \Gamma_+(t)A \Gamma_-(t' -s' - \{ z \}) v_l,
\end{multline*}
and thus, using that $\Gamma_+$ is the adjoint of $\Gamma_-$,
\begin{multline*}
[z^0] \left\langle \left( \tilde{A} \psi(z) \otimes \tilde{A} \psi^*(z) \right) \left( \Gamma_-(s') v_m \otimes \Gamma_-(-s') v_{l+1} \right) , \Gamma_-(s) v_{m+1} \otimes \Gamma_-(-s) v_{l} \right\rangle \\
=  [z^{l-m}]  \gamma(z^{-1},-2s')\tilde{\tau}_{m+1}(t+s,t'+s'+\{z\};A) \tilde{\tau}_l(t-s,t'-s'-\{z\} ; A ) .
\end{multline*}
Equality
\begin{multline}
[z^0] \left\langle \left( \psi(z) \tilde{A} \otimes \psi^*(z) \tilde{A}\right) \left( \Gamma_-(s') v_m \otimes \Gamma_-(-s') v_{l+1} \right) , \Gamma_-(s) v_{m+1} \otimes \Gamma_-(-s) v_{l} \right\rangle \\
= [z^{m-l}] \gamma(z^{-1},2s) \tilde{\tau}_m (t+s -\{z\} , t +s' ; A) \tilde{\tau}_{l+1}(t-s+ \{z \}, t-s'; A )
\end{multline}
is proved in a similar way, using that $z^C$ is self-adjoint and that $R$ is unitary.
\end{proof}

\section{Conclusion of the proof of Theorem \ref{thm:thm1}} \label{section:proof}

We use Proposition \ref{prop:prop2} to reduce the case to multiplicative functionals of unmodified Schur measures. We prove in Lemmas \ref{lem:commutation} and \ref{lem:expansion} that multiplicative functionals of Schur measures fit into the setting of Proposition \ref{prop:schurtau}.\\

Consider the operator
\begin{align*}
A_\sigma := \prod_{ k \in \Z'} (I-\sigma(k) \psi_k \psi_k^*). 
\end{align*}
The operator $A_\sigma$ is well defined and acts diagonally on the basis vector $v_S$ by
\begin{align} \label{eq:Asigma1}
    A_\sigma v_S = \prod_{x \in S} (1- \sigma(x)) v_S,
\end{align}
and we have the expansion
\begin{align} \label{eq:Asigmaexpansion}
    A_\sigma = \sum_{m \geq 0} (-1)^m \sum_{\{k_1, \cdots , k_m\} \subset \Z'} \sigma(k_1) \cdots \sigma(k_m)\psi_{k_1}\psi_{k_1}^* \cdots \psi_{k_m} \psi_{k_m}^*.
\end{align}
We start by the following Lemma, showing that the operator $A_\sigma$ satisfies the commutation condition of Proposition \ref{prop:schurtau}.
\begin{lem} \label{lem:commutation}
The operator $A_\sigma \otimes A_\sigma$ commutes with the operator $\Psi$, defined in (\ref{eq:defPsi}).
\end{lem}
\begin{proof}
Let $S,S' \subset \Lambda(\Z')$ be two semi-infinite subsets of $\Z'$. It suffices to prove that, for any $k_0 \in \Z$, we have
\begin{align*}
    (A_\sigma\otimes A_\sigma \cdot\psi_{k_0} \otimes \psi_{k_0}^*)(v_S \otimes v_{S'})=(\psi_{k_0}\otimes \psi_{k_0}^*\cdot A_\sigma\otimes A_\sigma) (v_S \otimes v_{S'} ).
\end{align*}
If $k_0 \in S$ or $k_0 \notin S'$, since $A_\sigma$ acts diagonally on the basis formed by the $v_S$, both sides of this equality are zero. Assume now that $k_0 \notin S$ and $k_0 \in S'$. From \eqref{eq:Asigma1} and \eqref{eq:defpsi}, \eqref{eq:defpsi*}, we see that
\begin{align*}
    A_\sigma\psi_{k_0}v_S=(1-\sigma(k_0))\psi_{k_0} A_\sigma v_S,
\end{align*}
while
\begin{align*}
    \psi_{k_0}^* A_\sigma v_{S'}=(1-\sigma(k_0))A_\sigma \psi_{k_0}^* v_{S'}.
\end{align*}
We deduce that
\begin{align*}
  (A_\sigma\otimes A_\sigma \cdot\psi_{k_0} \otimes \psi_{k_0}^*)(v_S \otimes v_{S'})&=(A_\sigma \psi_{k_0} v_S) \otimes (A_\sigma \psi_{k_0}^* v_{S'}) \\
  &=(1-\sigma(k_0)) (\psi_{k_0} A_\sigma v_S \otimes A_\sigma\psi_{k_0}^*v_{S'}) \\
  &=(\psi_{k_0} A_\sigma v_S )\otimes (\psi_{k_0}^*A_{\sigma} v_{S'}) \\
  &=(\psi_{k_0}\otimes \psi_{k_0}^*\cdot A_\sigma\otimes A_\sigma) (v_S \otimes v_{S'} ),
\end{align*}
and the lemma is proved.
\end{proof}
We can now establish that multiplicative functionals of Schur measures are the functions $\tilde{\tau}_n(t,t';A)$ of Proposition \ref{prop:schurtau} corresponding to $A=A_\sigma$.
\begin{lem} \label{lem:expansion}
    We have
\begin{align*}
\langle \Gamma_+( t) A_\sigma \Gamma_-(t') v_n ,v_n \rangle = \tau_n(t,t';\sigma).
\end{align*}
\end{lem}
\begin{proof}
Since the shift operator $R$ commutes withe vertex operators, we have from the expansion \eqref{eq:Asigmaexpansion} that
\begin{multline} \label{eq:fredholmsigma}
\langle \Gamma_+( t) A_\sigma \Gamma_-(t') v_n ,v_n \rangle = \langle \Gamma_+(t) \prod_{ k \in \Z'} (1-\sigma(k-n)) \psi_{k-n} \psi_{k-n}^* \Gamma_-(t') v_\emptyset, v_{\emptyset} \rangle \\
= \sum_{m \geq 0} (-1)^m \sum_{ \{x_1, \dots,x_m \} \subset \Z'} \sigma(x_1-n) \cdots  \sigma(x_m-n ) \langle \Gamma_+(t)  \psi_{x_1} \psi_{x_1}^* \cdots \psi_{x_m} \psi_{x_m}^* \Gamma_-(t') v_\emptyset , v_\emptyset\rangle.
\end{multline}
From Proposition \ref{prop:schurpsi} and Theorem \ref{thm:okounkov}, we have
\begin{align*}
\langle \Gamma_+(t)  \psi_{x_1} \psi_{x_1}^* \cdots \psi_{x_m} \psi_{x_m}^* \Gamma_-(t') v_\emptyset , v_\emptyset\rangle = Z_{t,t'} \det \left( K_{t,t'}  (x_i,x_j) \right)_{i,j=1}^m.
\end{align*}
Thus, one recognizes on the last line of equation (\ref{eq:fredholmsigma}) the Fredholm determinant expansion of
\begin{align*}
Z_{t,t'} \det \left( 1- \sigma( \cdot - n ) K_{t,t'}  \right)_{\ell^2(\Z')} =Z_{t,t'} \E_{\P_{t,t'}} \left[ \prod_{x \in X} (1- \sigma(x-n) )\right].
\end{align*}
We conclude by Proposition \ref{prop:prop2}.
\end{proof}
We now conclude the proof of Theorem \ref{thm:thm1} by applying Proposition \ref{prop:schurtau} to $A= A_\sigma$.

\end{document}